\documentclass[12pt]{amsart}
\pdfoutput=1
\usepackage[utf8]{inputenc}
\usepackage{MnSymbol}
\usepackage{accents}
\usepackage{tikz-cd}

\usepackage{bm}
\usepackage[textsize=footnotesize,shadow]{todonotes}
\usepackage{fancyhdr}
\usepackage[colorlinks=true,unicode=true]{hyperref}
\usepackage{mathtools}
\usepackage{lpic}
\input{style.tex}
\input{defs.tex}

\def\draft{\fbox{\tiny draft: \today}}
\let\draft=\relax
\setcolors[rmcolor=cyan,
           commentcolor=purple,
           inscolor=blue,
           commentsymbol=$\boxed{\mathbb{S}}$,
           rmsymbol=$\boxed{\mathbb{S}_{rm}}$,
           inssymbol=$\boxed{\mathbb{S}_{ins}}$]
          {s}
\setcolors[rmcolor=green,
           commentcolor=red,
           inscolor=brown,
           commentsymbol=$\boxed{\mathbb{J}}$,
           rmsymbol=$\boxed{\mathbb{J}_{rm}}$,
           inssymbol=$\boxed{\mathbb{J}_{ins}}$]
          {j}

\title{Tropical probability theory and an application to the entropic cone}
\author[RM]{R. Matveev}
\author[JWP]{J. W. Portegies}
\begin{document}
\thispagestyle{fancy} 
\begin{abstract}
In a series of articles, we have been developing a theory of
\emph{tropical diagrams of probability spaces}, expecting it to be
useful for information optimization problems in information theory and
artificial intelligence. In this article, we give a summary of our
work so far and apply the theory to derive a dimension-reduction
statement about the shape of the entropic cone.
\end{abstract}

\maketitle

%\tableofcontents
\son
\jon

\section{Introduction}
With the aim of developing a systematic approach to an important class
of problems in information theory and artificial intelligence, we
started in \cite{Matveev-Asymptotic-2018} the development of a theory
of \emph{tropical diagrams of probability spaces}. One of our intended
applications is to characterize, or at least derive important
properties of, the \emph{entropic cone}: an important open problem in information theory, to which Fero  Mat\'u\v s made invaluable contributions.

In this article, we give a summary of our work on tropical diagrams so
far and apply the technology to derive a statement about
the entropic cone.

We briefly recall the definition of the entropic cone. 
Given a collection of $k$ random
variables $\Xsf_1, \dots, \Xsf_k$ and a subset $I \subset \{ 1, \dots,
k \}$, we can record the entropy of the joint random variable
$\Xsf_I$. This way, we get a function from $\Lambdabf_k$ to $\Rbb$,
where $\Lambdabf_k$ denotes the set of nonempty subsets of $\{
1, \dots, k\}$. We interpret the function as an element of the vector
space $\Rbb^{\Lambdabf_k}$ and call it the entropy vector of the
random variables $\Xsf_1, \dots, \Xsf_k$. In general, we say that a
vector in $\Rbb^{\Lambdabf_k}$ is entropically representable if it is
the entropy vector of some collection of random variables
$\Xsf_1, \dots, \Xsf_k$.

The entropic cone is the closure of the set of entropically
representable vectors. Entropies of random variables, conditional
entropies, mutual information and conditional mutual information are
all nonnegative. These conditions are called the Shannon
inequalities. For $k \leq 3$, the Shannon inequalities completely
describe the entropic cone, but for $k\geq 4$ the situation is much
more complicated. Zhang and Yeung showed that the entropic cone and
the submodular cone (i.e. the cone cut out by Shannon inequalities)
are different \cite{Zhang-Characterization-1998}, by identifying a
non-Shannon inequality satisfied by all entropically representable
vectors. Subsequently, more non-Shannon inequalities were discovered,
e.g.~\cite{Makarychev-New-2002, Dougherty-Six-2006}.

In \cite{Matus-Infinitely-2007} Mat\'u\v s discovered several infinite
families of linear inequalities satisfied by the entropic cone and
used it in a clever way to show that the cone is \emph{not
polyhedral}.
Other infinite families of \emph{information inequalities} were
found in~\cite{Dougherty-Non-Shannon-2011} as well as many (more than
200) sporadic inequalities.

In the case of four random variables, the entropic cone is a closed
convex cone in $\Rbb^{15}$. Using techniques developed
in~\cite{Matveev-Asymptotic-2018, Matveev-Tropical-2019,
Matveev-Conditioning-2019, Matveev-Arrow-2019}, we show how the
dimension of the problem of determining the entropic cone could be
reduced from 15 to 11.

During our work on the development of tropical probability we
were greatly influenced by the article of
Gromov \cite{Gromov-Search-2012} and by numerous discussions with Fero
Mat\'u\v s as well as by his published
work, such as \cite{Matus-Probabilistic-1993, Matus-Conditional-1995,
Matus-Infinitely-2007}.

\section{Tropical Diagrams}
The language of random variables was introduced by Fr\'echet,
Kolmogorov and others, so that joint distributions are automatically
defined. For our purposes, this is not a convenient setup, as we often
need to vary the joint distributions. That's why we use a different
language of diagrams of probability spaces, which we introduce
below. A more detailed discussion and proofs of the statements below
can be found in~\cite{Matveev-Asymptotic-2018},%
~\cite{Matveev-Tropical-2019} and~\cite{Matveev-Arrow-2019}.
  
\subsection{Probability spaces}
For the purposes of this article, a \term{probability space} is a set
with a probability measure on it which is supported on a finite
subset. A \term{reduction} from one probability space to another is an
equivalence class of measure-preserving maps, where two maps are
considered equivalent if they coincide on a set of full measure. Note
that the target space of a random variable taking values in a finite
set is a probability space according to this definition.

The \term{tensor product} $X \otimes Y$ of two probability spaces $X$
and $Y$ is the independent product.

\subsection{Diagrams of probability spaces}
We will consider commutative diagrams of probability spaces and
reductions, such as a \term{two-fan} and a \term{diamond}, pictured
below
\[
\tageq{fan-and-diamond}
\begin{cd}[row sep=tiny, column sep=small]
\mbox{}
\& 
Z
\arrow{dl}
\arrow{dr}
\&
\\
X
\&
\&
Y
\end{cd}
\quad\quad
\begin{cd}[row sep=-1mm,column sep=small]
\mbox{}
\& 
Z
\arrow{dl}
\arrow{dr}
\&
\\
X
\arrow{dr}
\&
\&
Y
\arrow{dl}
\\
\&
W
\&
\end{cd}
\]
In these diagrams, $X$, $Y$, $Z$ and $W$ are probability spaces, and
the arrows are reductions. To speak about general diagrams, we will
need to specify the arrangement of probability spaces and reductions,
i.e.~we need to record the underlying combinatorial structure. There
are several, equivalent, ways to do so: using a poset category, a
partially ordered set (poset), or a directed acyclic graph
(DAG) with some additional properties as described below. From
our perspective, the language of categories is most convenient for
this purpose, but it may not be as familiar as the other two
concepts. That is why we will provide a dictionary to convert from one
setup to the other.

\subsubsection{Categories, posets and DAGs}
A \term{poset category} is a finite category $\Gbf$ such that for any
pair of objects $i, j \in \Gbf$ there is at most one morphism either
way. 
We will require the poset categories used for indexing
diagrams to have an additional property, that we describe below after
introducing some convenient terminology.

For a morphism $i\to j$ in $\Gbf$, the object $i$ will be called
an \term{ancestor} of $j$ and object $j$ will be called a \term{descendant}
of $i$.

An \term{indexing category} $\Gbf$ is a finite poset category such
that for any two objects $i,j\in\Gbf$ there exists a \term{minimal
common ancestor $\hat\imath$}, that is an object $\hat\imath$ wich is
an ancestor to both $i$ and $j$ and such that any other common
ancestor of $i$ and $j$ is also an ancestor of $\hat\imath$.
For an interested reader an example of a poset category
that fails this property is shown below.
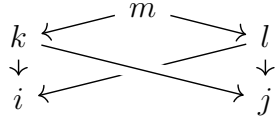
\[
\begin{cd}[row sep=-1.9mm]
  \mbox{}
  \&
  m
  %\arrow{dl}
  \arrow{dr}
  \mbox{}
  \\
  k
  \arrow{dd}
  \&
  \mbox{}
  \&
  l
  \arrow{dd}
  \arrow{ddll}
  \\
  \mbox{\rule{0mm}{3mm}}
  \\
  i
  \&\&
  j
  \arrow[from=uull,crossing over]
  \arrow[from=uuul,to=uull]
\end{cd}
\]

Given a poset $(P, \geq)$ such that any subset in $P$ has a supremum
(a least common upper bound), one can construct an indexing
category $\Gbf$, having as objects the points in the poset, and a
unique morphism $i \to j$ for any pair $i \geq j$.

Starting with a DAG, one can construct a
poset category by taking the transitive closure of the DAG and
considering vertices as objects and arrows as morphisms.
The translation of the defining property of indexing categories is
straightforward in the DAG language. 

A \term{fan} in a category is a pair of morphisms with the same domain 
$(i\ot k\to j)$. Such a fan is called \term{minimal} if whenever it is
included in a commutative diagram
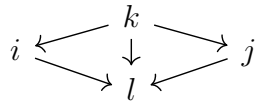
\[       
\begin{cd}[row sep=-1.5mm]
  \mbox{}\&
  k
  \arrow{dl}
  \arrow{dd}
  \arrow{dr}
  \&
  \\
  i
  \arrow{dr}
  \&\mbox{}\&
  j
  \arrow{dl}
  \\
  \&
  l
\end{cd}
\]       
the verical arrow $k\to l$ must be an isomorphism.

Indexing categories have the following useful properties, which are
elementary to establish. First, for any pair of objects $i,j$ in an
indexing category $\Gbf$, there exists a \term{unique minimal fan}
in $\Gbf$ with target objects $i$ and $j$. Secondly, any indexing
category is \term{initial}, i.e. it has
an \term{initial object} that is an ancestor to
any other object in $\Gbf$.

\subsubsection{Diagrams} 
A \term{diagram of probability spaces} is a functor $\Xcal$ from an
indexing category $\Gbf = \{i; \gamma_{ij}\}$ to the category of
probability spaces. Essentially, this means that given an indexing
category, poset or DAG, we get a $\Gbf$-diagram of probability spaces
$\Xcal=\{X_i; \chi_{ij}\}$ by assigning to each object/vertex $i$ a
probability space $X_i$ and to each morphism/arrow $\gamma_{ij}$ a
reduction $\chi_{ij}$, requiring that the resulting diagram
commutes. We denote the set of all $\Gbf$-diagrams of probability
spaces by $\prob\langle \Gbf \rangle$.

\subsubsection{Full diagrams and random variables}
\label{se:full}
Important examples of diagrams are $\Lambdabf_n$-diagrams, which we
call \term{full diagrams}, where $\Lambdabf_n$ is the poset of
non-empty subsets of the set $\{1, \dots, n\}$ ordered by
inclusion. Given an $n$-tuple of random variables
$(\Xsf_1, \dots, \Xsf_n )$ we can construct a $\Lambdabf_n$-diagram
\[
  \< \Xsf_1, \dots, \Xsf_n \> := \{ X_I ; \chi_{IJ} \}
\]
by setting $X_I$ equal to the target space of $(X_i : i \in I)$ with
the induced distribution and $\chi_{IJ}$ equal to the natural
projections.  On the other hand, starting with a $\Lambdabf_n$-diagram
we can construct an $n$-tuple of random variables as reductions from
the initial space to $n$ terminal spaces.
Diagrams of combinatorial type $\Lambdabf_{2}$ are two-fans,
  pictured above in (\ref{eq:fan-and-diamond}), and $\Lambdabf_{1}$-diagrams are single probability
  spaces.

\subsubsection{Diagrams of diagrams}
A \term{reduction} $\rho:\Xcal \to \Ycal$ from a $\Gbf$-diagram
$\Xcal$ to a $\Gbf$-diagram $\Ycal$ is a natural transformation from
(the functor) $\Xcal$ to $\Ycal$. It amounts to specifying a reduction
$\rho_i : X_i \to Y_i$ for every $i$, such that the diagram
obtained from $\Xcal$, $\Ycal$ and the $\rho_i$'s is
commutative. Thus, $\prob\langle \Gbf \rangle$ is itself a category.

Hence, we can also construct diagrams of diagrams. Most important for
us are two-fans of $\Gbf$ diagrams,
\[
\begin{cd}[row sep=0mm]
\mbox{}
\& 
\Zcal
\arrow{dl}
\arrow{dr}
\&
\\
\Xcal
\&
\&
\Ycal
\end{cd}
\]
where $\Xcal$, $\Ycal$ and $\Zcal$ are $\Gbf$-diagrams, and the arrows
are reductions of diagrams.  For the space of $\Hbf$-diagrams of
$\Gbf$-diagrams we will use the notation $\prob\<\Gbf\>\<\Hbf\> =
\prob\<\Gbf, \Hbf\>$. Note that an $\Hbf$-diagram of $\Gbf$-diagrams
can equivalently be interpreted as a $\Gbf$-diagram of
$\Hbf$-diagrams.

\subsubsection{Minimal diagrams}
A $\Gbf$-diagram  $\Xcal$ is called \term{minimal}
if it maps minimal fans in $\Gbf$ to minimal fans in the target category.

A \term{minimization} of a two-fan $\hat\Zcal := (\Xcal \ot
  \Zcal \to \Ycal)$ of either of probability spaces or of diagrams
is the minimal fan
$\hat \Ccal$ and a reduction 
\[
\tageq{minimal}
\begin{cd}[row sep=5mm]
\Xcal
\arrow{d}{f}
\&
\Zcal
\arrow{l}
\arrow{d}{h}
\arrow{r}
\&
\Ycal
\arrow{d}{g}
\\
\Acal
\&
\Ccal
\arrow{l}
\arrow{r}
\&
\Bcal
\end{cd}
\]
such that $f$ and $g$ are isomorphisms.

It is shown in \cite[Proposition 2.1]{Matveev-Asymptotic-2018} that a
minimization always exists and is unique up to isomorphism.

\skipthis{
A general $\Gbf$-diagram will be called \term{minimal} if for any pair
of distinct spaces in it, it also contains a minimal fan with the
spaces as feet.  We will denote the space of minimal $\Gbf$-diagrams
by $\prob\<\Gbf\>_\msf$.
}

We will also refer to a minimal two-fan with $\Xcal$ and $\Ycal$ as
targets, as a coupling between $\Xcal$ and $\Ycal$.

\subsubsection{Tensor product and conditioning}
The \term{tensor product} of two $\Gbf$-diagrams $\Xcal=\{X_i;
\chi_{ij}\}$ and $\Ycal =\{Y_i; \upsilon_{ij}\}$ is $\Xcal \otimes
\Ycal := \{X_i\otimes Y_i; \chi_{ij}\times\upsilon_{ij}\}$.

If $\Xcal$ is a $\Gbf$-diagram, and $U$ is a probability space in
$\Xcal$, then the whole diagram $\Xcal$ can be conditioned on an
outcome $u \in U$ with positive weight. We denote the conditioned
diagram by $\Xcal \rel u$. A precise definition of this construction
is given in \cite[Section 2.8]{Matveev-Asymptotic-2018}.

\subsection{The intrinsic and asymptotic entropy distances}
For a given a $\Gbf$-diagram $\Xcal$ we may evaluate the entropies of the
individual probability spaces, which gives a map
\[
\ent_* : \prob\<\Gbf\> \to \Rbb^\Gbf
\]
where the target space $\Rbb^\Gbf$ is the vector space of all
real-valued functions on the set of objects in $\Gbf$, equipped with
the $\ell^1$-norm. The entropy is a homomorphism in the sense that
$\ent_*(\Xcal \otimes \Ycal) = \ent_*(\Xcal) + \ent_*(\Ycal)$.

Given a two-fan of $\Gbf$-diagrams $\Kcal = (\Xcal \ot \Zcal \to
\Ycal)$, the \term{entropy distance} between $\Xcal$ and $\Ycal$ is
defined by
\[
  \kd(\Kcal) 
  := 
  \| \ent_*(\Zcal) - \ent_*(\Xcal) \|_1 + 
  \| \ent_*(\Zcal) - \ent_*(\Ycal) \|_1
\]
We use the entropy distance as a measure of deviation of a fan
$\Kcal$ from being an isomorphism between $\Xcal$ and $\Ycal$. Indeed,
the entropy distance $\kd(\Kcal)$ vanishes if and only if both arrows
in $\Kcal$ are isomorphims.

We obtain the \term{intrinsic entropy distance} $\ikd(\Xcal, \Ycal)$
between two $\Gbf$-diagrams $\Xcal$ and $\Ycal$ by taking an infimum
over all couplings between $\Xcal$ and $\Ycal$
\[
  \ikd(\Xcal, \Ycal) 
  := 
  \inf \{ \kd(\Kcal) : 
  \ \Kcal \text{ coupling between } \Xcal \text{ and } \Ycal \}
\]
For probability spaces, the intrinsic entropy distance was introduced in \cite{Kovavcevic-Hardness-2012, Vidyasagar-Metric-2012}.

We also define the \term{asymptotic entropy distance} $\aikd(\Xcal,
\Ycal)$ by
\[
  \aikd(\Xcal, \Ycal) 
  := 
  \lim_{n \to \infty} \frac{1}{n} \ikd(\Xcal^n, \Ycal^n)
\]
where $\Xcal^n$ denotes the $n$-fold independent product of $\Xcal$. 

Both $\ikd$ and $\aikd$ are pseudo-distance functions in that they
satisfy all axioms of a distance function, except that they may vanish
on pairs of non-identical points, see
~\cite{Matveev-Asymptotic-2018} and~\cite{Matveev-Tropical-2019} for the proofs.  

\subsection{Tropical diagrams}
Tropical objects, as for instance encountered in algebraic geometry,
are, roughly speaking, divergent sequences of classical objects
(e.g.~algebraic varieties), renormalized by viewing them on a log
scale with increasing base.

The space of tropical diagrams of probability spaces is defined along
similar lines: it consists of certain divergent sequences of
diagrams and is endowed with an asymptotic entropy distance, thus
achieving a similar renormalization.

Our description below is extremely brief. For details and proofs, we
refer the reader to \cite{Matveev-Tropical-2019}.

\subsubsection{Quasi-linear sequences}
We define the \term{linear sequence} generated by a $\Gbf$-diagram
$\Xcal$ as the sequence $\bernoulli{\Xcal} := (\Xcal^n : n \in
\Nbb_0)$ and we define the distance between two such sequences by
\[
  \aikd\Big(\bernoulli{\Xcal}, \bernoulli{\Ycal}\Big) 
  := 
  \lim_{n \to \infty} \frac{1}{n}\ikd\Big(\Xcal^n, \Ycal^n\Big)
\]
Tropical diagrams of probability spaces will be sequences that are
almost linear, so that it allows us to define algebraic operations on
them, and establish completeness of the space of all tropical
diagrams.

We call a sequence $[\Xcal]:= (\Xcal(n): n \in \Nbb_0)$ \term{quasi-linear}
if for every $m, n \in \Nbb$,
\[
\aikd\Big(\Xcal(m+n), \Xcal(m) \otimes \Xcal(n) \Big) \leq C (m + n)^{3/4}
\]
and define the distance between two quasi-linear sequences by 
\[
  \aikd\Big( [\Xcal], [\Ycal] \Big) 
  := \lim_{n \to \infty} \frac{1}{n} \aikd\Big(\Xcal(n), \Ycal(n)\Big)
\]
and denote by $\prob[\Gbf]$ the (pseudo-)metric space of all
quasi-linear sequences endowed with $\aikd$. Two quasi-linear
sequences will be called asymptotically equivalent if they are zero
distance apart.  Equivalence classes of quasi-linear sequences will be
called \term{tropical diagrams of probability spaces}.  In our
discussions we will sometimes be sloppy, and {make no distinction
  between equivalence classes and their representatives (quasi-linear
  sequences)}. This is harmless, as operations we consider are all
$\aikd$-continuous and preserve asymptotic equivalence.

The sum of two sequences is defined as element-wise tensor product,
and multiplication by a scalar $\lambda \geq 0$ is defined by
\[
\lambda \cdot [\Xcal] 
:= ( \Xcal( \lfloor \lambda \cdot n \rfloor ) : n \in \Nbb_0)
\]
The addition and scalar multiplication satisfy the usual associative,
commutative and distributive laws up to asymptotic
equivalence. Therefore, the space $\prob[\Gbf]$ has the structure of a
convex cone.

The asymptotic distance $\aikd$ is $1$-homogeneous
\[
  \aikd\Big( \lambda \cdot [\Xcal], \lambda \cdot [\Ycal]\Big) 
  = \lambda \aikd\Big( [\Xcal], [\Ycal]\Big)
\]
and translation-invariant,
\begin{align*}
  \aikd\Big( [\Xcal] + [\Zcal], [\Ycal] + [\Zcal] \Big) 
  &=\aikd\Big( [\Xcal], [\Ycal] \Big)
\end{align*}
We show in \cite{Matveev-Tropical-2019} that the space $\prob[\Gbf]$ is
complete. Together with the algebraic structure, it implies that
$\prob[\Gbf]$ is a closed convex cone in some (generally
infinite-dimensional) Banach space $B$. We call elements in the dual
space $B^*$ \term{entropic quantities}. The entropy functional defined by
\[
\tageq{tropical-entropy}
\ent_*([\Xcal]) := \lim_{n \to \infty} \frac{1}{n} \ent_*(\Xcal(n)) 
\]
yields an example of such dual elements.

\subsection{Homogeneous diagrams and asymptotic equipartition property}
We call a diagram of probability spaces $\Xcal$ \term{homogeneous} if
its automorphism group $\Aut(\Xcal)$ acts transitively on every space
in $\Xcal$.  

Examples of homogeneous diagrams can be constructed in
the following way. A \term{$\Gbf$-diagram
of groups} is a pair consisting of an ambient finite group $G$
and a $\Gbf$-diagram of subgroups of $G$, $(H_i : i \in \Gbf)$, where
the arrows are inclusions. Starting with a $\Gbf$-diagram of groups, we construct
a $\Gbf$-diagram of probability spaces $\Xcal := \{ X_i
; \chi_{ij} \}$ by setting $X_i := G / H_i$ with the uniform measure and defining the
reduction $\chi_{ij}$ to be the natural projection
$\chi_{ij} : G/H_i \to G/H_j$ whenever $H_i \subset H_j$. In fact,
every homogeneous diagram arises in this way \cite[Section
2.7.1]{Matveev-Asymptotic-2018}. We call a homogeneous
diagram \term{Abelian} if it can be constructed in this way with
Abelian $G$.

Starting with a diagram of groups $\set{G; H_{i}: i\in \Gbf}$ the
resulting homogeneous diagram will be minimal if and only if for any
$i,j\in\Gbf$ there exists $k\in\Gbf$, such that $H_{k}=H_{i}\cap
H_{j}$.  If a diagram of groups satisfies this property, we will call
it \term{minimal} as well. When $\Gbf$ is a full
indexing category $\Gbf=\Lambdabf_{n}$, the description of minimal diagrams of
groups is especially simple: One needs to specify only the terminal
groups, others being obtained by appropriate intersections. We will
write $(G;\;H_{1},\ldots,H_{n})$ for such minimal
$\Lambdabf_{n}$-diagram of groups.

We denote the space of quasi-linear sequences of homogeneous
$\Gbf$-diagrams by $\prob[\!\Gbf]_\hsf$, the subspace of sequences of
Abelian $\Gbf$-diagrams by $\prob[\Gbf]_{\mathsf{Ab}}$, and the space 
quasi-linear sequences of minimal $\Gbf$-diagrams by $\prob[\Gbf]_\msf$.

The Asymptotic Equipartition Property for diagrams of
  probability spaces that we have shown in
\cite[Theorem 6.1]{Matveev-Asymptotic-2018} essentially says that
  any linear sequences of diagrams of probability spaces is
asymptotically equivalent to a quasi-linear sequence of homogeneous
diagrams. Together with the density of linear sequences in
$\prob[\Gbf]$, ~\cite[Theorem 5.2]{Matveev-Tropical-2019}, it implies the
following theorem.

\begin{theorem}{p:aep-tropical}
  \ \ For any indexing category $\Gbf$ the spaces $\prob[\Gbf]_{\hsf}$ and
  $\prob[\Gbf]_{\hsf,\msf}$ are dense in $\prob[\Gbf]$ and
  $\prob[\Gbf]_{\msf}$, respectively.
\end{theorem}

Intuitively, according to the theorem above, one may think of a
tropical diagram as a homogeneous diagram of very large probability
spaces. Thus, whenever one wants to evaluate a continuous linear
functional on a diagram $\Xcal$, one may assume that it is homogeneous
and consists of arbitrarily large spaces. We take advantage of this
point of view in the next section.

As a trivial, but enlightening, example, consider a two-fan $(X \ot Z
\to Y)$.  By Theorem \ref{p:aep-tropical}, a high power $(X^n \ot Z^n
\to Y^n)$ can be approximated by a homogeneous fan $H_X \ot H_Z \to
H_Y$. The entropies of $X$, $Z$ and $Y$ (and therefore also the mutual
information $\ent(X) + \ent(Y) - \ent(Z)$ between $X$ and $Y$) can be
established by just counting in the homogeneous fan:
$\ent(X) \approx \frac{1}{n} \log |H_X|$ where $|H_X|$ denotes the
cardinality of $H_X$. However, the three entropies do not determine
the asymptotic equivalence class of the fan, there are many more
(in fact, infinitely many) independent entropic
  quantities. But all can be determined
from the homogeneous approximation.

\subsection{Tropical conditioning}
\subsubsection{Conditioning}
One of the advantages of homogeneous diagrams is that if a homogeneous
diagram $\Xcal$ contains a space $U$, then the isomorphism class of
the conditioned diagram $\Xcal\rel u$ does not depend on the choice of
an atom $u\in U$.
  
Since any tropical diagram is asymptotically equivalent to a
homogeneous tropical diagram, we can use this independence of $u$ to
define an operation of conditioning of a tropical $\Gbf$-diagram
$[\Xcal]$ on a space $[U]$ in it, obtaining another tropical
  $\Gbf$-diagram denoted by $[\Xcal \rel U]$.  In the tropical
setting the diagram $[\Xcal\rel U]$ depends
(Lipschitz-)continuously on $[\Xcal]$ and
$[U]$. This subject is discussed in more details
in~\cite{Matveev-Arrow-2019}. 
  
\subsubsection{Entropy and mutual information}
\Label{s:entropy}
Now that we defined $[\Xcal\rel U]$ as a tropical diagram, its entropy
$\ent_*([\Xcal \rel U])$ is defined by the limit
in~(\ref{eq:tropical-entropy}). At the same time, it equals the limit
\[
  \lim_{n \to \infty} \frac{1}{n} \ent_*( \Xcal(n) \rel U(n) ) := 
  \lim_{n \to \infty} \frac{1}{n} \int_{u \in U(n)} \ent_*(\Xcal(n) \rel u) \d p(u)
\]
  
In~\cite{Matveev-Tropical-2019} it is shown that the
  space of single tropical probability spaces, $\prob[\!\Lambdabf_{1}\!]$,
is isomorphic to $\Rbb_{\geq 0}$, with the
isomorphism given by the entropy. Thus, a tropical probability space is
completely determined by its entropy, and we will simply write
\begin{itemize}
\item 
  $[X]$ for $\ent([X])$.
\item 
  $[X : Y]:= [X] + [Y] - [Z]$ for the mutual information between $X$
  and $Y$ in the minimal two-fan $[X] \ot[] [Z] \to[] [Y]$
\item 
  $[X : Y\rel U]:= [\hat X] + [\hat Y] - [\hat Z] -[U]$ for the
  conditional mutual information between $X$ and $Y$, where $[\hat
    X]$, $[\hat Y]$, $[\hat Z]$ and $[U]$ are the spaces in the
  minimal diagram
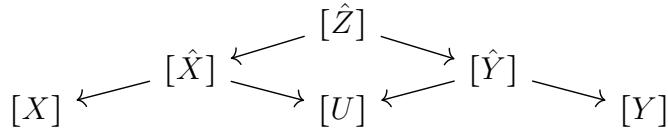
\[
\begin{cd}[column sep=normal, row sep=-2mm]
  \mbox{}
  \&{}
  \&{}
    [\hat Z]
    \arrow{dl}
    \arrow{dr}
  \&{}
  \&{}
  \\{}
  \&{}
    [\hat X]
    \arrow{dl}
    \arrow{dr}
  \&{}
  \&{}
    [\hat Y]
    \arrow{dr}
    \arrow{dl}
  \&{}
  \\{}
    [X]
  \&{}
  \&{}
    [U]
  \&{}
  \&{}
    [Y]
\end{cd}
\]
\end{itemize}

\section{Arrow Contraction and Expansion}
In this section we describe two operations on tropical diagrams,
\term{arrow contraction and expansion}. Given a tropical diagram
$[\Zcal]$, arrow contraction is a modification of the diagram in such a
way that a certain arrow becomes an isomorphism, while keeping control
of what happens to some other parts of $[\Zcal]$. Arrow expansion is
an inverse operation. We will apply these techniques in the
next section to derive a dimension reduction result for the entropic
cone. 

The full construction is quite involved. Here we will only describe a
corollary necessary for our purposes, and refer the reader to
\cite{Matveev-Arrow-2019} for the full results and details.

\subsection{Admissible and reduced sub-fans}
Suppose $\Zcal$ is a $\Gbf$-diagram and $X$ is an element in it.  By
the \term{ideal} generated by $X$ we wean the sub-diagram $\lc X\rc$
of $\Zcal$, that consists of the target spaces of all morphisms
starting in $X$ and all (available in $\Zcal$) arrows between them. We
will sometimes refer to spaces in $\lc X\rc$ as the \term{descendants}
of $X$. The ideal generated by a space $X$ included in some diagram will
be denoted $\lc X\rc$.

If $[\Zcal]$ is a diagram of tropical probability spaces with the
tropical space $[X]$ in it, in order to unclutter notations we will
write
\[
\lc X \rc := \lc \rule{0mm}{2ex}[X] \rc
\]

An \term{admissible sub-fan} $(X\ot Z\to U)$ in a diagram $\Zcal$ is
a \emph{minimal} sub-fan such that the space $U$ is terminal, i.e.~it
is not the domain of definition of any (non-identity) morphism in
$\Zcal$. An admissible fan will be called reduced if $Z\to X$ is an
isomorphism.

A diagram with an admissible fan is illustrated schematically in
Figure~\ref{fig:lowering}. Two more concrete examples
are shown in Figures~\ref{fig:lowering-in-L2}
and~\ref{fig:lowering-in-L3}.

\usetikzlibrary{decorations.pathmorphing} 
\begin{figure}[h]
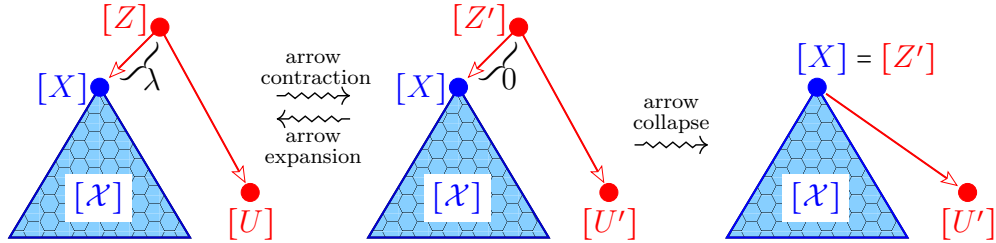

  \begin{tabular}{ccccc}
    \begin{lpic}[clean,r(-5mm)]{adjusted(0.4)}
      \lbl[W]{32,14;\textcolor{blue}{$[\Xcal]$}}
      \lbl[r]{27,50;\textcolor{blue}{$[X]$}}
      \lbl[r]{47,73;\textcolor{red}{$[Z]$}}
      \lbl[bl]{36,55,-45;$\left.\rule{0mm}{4.5mm}\right\}$}
      \lbl[tl]{45,57;$\lambda$}
      \lbl[t]{81,10.5;\textcolor{red}{$[U]$}}
    \end{lpic}
    &
    \begin{cd}[row sep=0mm]
     \mbox{}
     \arrow[rightsquigarrow]{r}{%
       \parbox{14mm}{\tiny\center arrow\\[-0.6ex] contraction\\[0.5ex]}}
     \&
     \mbox{}
     \\
     \mbox{}
     \&
     \arrow[rightsquigarrow]{l}{%
       \parbox{14mm}{\tiny\center arrow\\[-0.6ex] expansion\\[0.5ex]}}
     \\
     \rule{0mm}{30mm}
     \&
     \mbox{}
    \end{cd}
    &
    \begin{lpic}[draft,clean,l(-5.1mm),r(-5.1mm)]{adjusted(0.4)}
      \lbl[W]{32,14;\textcolor{blue}{$[\Xcal]$}}
      \lbl[r]{27,50;\textcolor{blue}{$[X]$}}
      \lbl[r]{47,73;\textcolor{red}{$[Z']$}}
      \lbl[bl]{36,55,-45;$\left.\rule{0mm}{4.5mm}\right\}$}
      \lbl[tl]{45,57;$0$}
      \lbl[t]{81,10.5;\textcolor{red}{$[U']$}}
    \end{lpic}
    &
    \begin{cd}[row sep=20mm]
     \mbox{}
     \arrow[rightsquigarrow]{r}{%
       \parbox{12mm}{\tiny\center arrow\\[-0.6ex] collapse\\[0.5ex]}}
     \&
     \mbox{}
     \\
     \mbox{}
    \end{cd}
    &
    \begin{lpic}[clean,l(-5mm)]{adjusted1(0.4)}
      \lbl[W]{32,14;\textcolor{blue}{$[\Xcal]$}}
      \lbl[b]{47,55;$\textcolor{blue}{[X]}=\textcolor{red}{[Z']}$}
      \lbl[t]{81,10.5;\textcolor{red}{$[U']$}}
    \end{lpic}
  \end{tabular}
  \caption{Arrow contraction/expansion and arrow collapse. 
    \emph{Here $[\Xcal]:=\lc X\rc$. In
      the left diagram the fan $[X]\ot{} [Z]\to{}[U]$ is
      admissible. In the middle diagram  $[X]\ot{} [Z']\to{}[U']$ is
      admissible and reduced. The diagrams may contain some
      other spaces beyond those shown. During contraction/expansion we
      don't have control over the other parts of the diagram.}}
  \label{fig:lowering}
\end{figure}

If an arrow $Y\to{}X$ in a diagram $\Zcal$ is an isomorphism, then we
can identify the spaces $X$ and $Y$, thus changing the combinatorial
structure of $\Zcal$. We call such change \term{arrow collapse}.
Examples of the process of collapsing an arrow can be seen in
Figures~\ref{fig:lowering},~\ref{fig:lowering-in-L2}
and~\ref{fig:lowering-in-L3}.  

\subsection{Arrow contraction and expansion.}
Suppose $[\Zcal]$ and $[\Zcal']$ are two tropical $\Gbf$-diagrams, containing admissible fans $([X]\ot{}[Z]\to{}[U])$
and $([X']\ot{}[Z']\to{}[U'])$, respectively, which correspond to each
other under the combinatorial isomorphism between $[\Zcal]$ and
$[\Zcal']$. Suppose the fan $([X']\ot{}[Z']\to{}[U'])$ is
reduced. Denote $[\Xcal]:=\lc X\rc$ and $[\Xcal']:=\lc X'\rc$. 
 
We say that
  \emph{$[\Zcal']$ is obtained from $[\Zcal]$ by arrow contraction}
  or, alternatively,
\emph{
  $[\Zcal]$ is obtained from $[\Zcal']$ by arrow expansion},
if 
\begin{align*}
  \tageq{lowering-equality}
  [\Xcal]&=[\Xcal']\\
  \tageq{lowering-cond-equality}
  [\Xcal\rel U]&=[\Xcal'\rel U']
\end{align*}

The other spaces in $[\Zcal]$ outside of $[\Xcal]$ and $[U]$ may
change in an uncontrolled manner.  
In view of equality~(\ref{eq:lowering-equality}) we will identify
diagrams $[\Xcal]$ and $[\Xcal']$.

Arrow contraction, expansion and collapse are illustrated in
Figures~\ref{fig:lowering}, \ref{fig:lowering-in-L2} and
\ref{fig:lowering-in-L3}.

Note that equation~(\ref{eq:lowering-cond-equality}) is in general an equality
in an infinite-dimensional space. But as a simple
consequence we have that for any two spaces $[X_1]$ and $[X_2]$ in
$[\Xcal]$ and the corresponding spaces $[X'_{1}]$ and $[X'_{2}]$ in
$[\Xcal']$ the following equalities hold:
\begin{align*}
  [X_1\rel U]&=[X'_1\rel U']\\
  [X_1:X_2\rel U]&=[X'_1:X'_2\rel U']\\
  [X_1:U]&=[X'_1:U']\tageq{entropy-preservation}\\
  [X_1:U\rel X_2]&=[X'_1:U'\rel X'_2]\\
  [U']&=[X:U]
\end{align*}
Indeed, the first equality follows directly from
equality~(\ref{eq:lowering-cond-equality}). The next three can be
  proven by expanding the right- and left-hand sides into summands of
  the form $[A\rel U]$ and $[B]$ for some $[A]$ and $[B]$ in
  $[\Xcal]$. The last one follows from the fact that $[U']$ is a
  descendant of $[X]$ in $[\Zcal']$ and therefore $[X\rel U']=[X]-[U']$.

We expect that arrow contraction is possible for any diagram with an
admissible fan, but for the purposes of this article the following
approximate contraction result from~\cite{Matveev-Arrow-2019}
suffices.

\begin{theorem}{p:contraction}
  Let $([X]\ot{}[Z]\to{}[U])$ be an admissible fan in some tropical
  $\Gbf$-diagram $[\Zcal]$. Then for every $\epsilon>0$ there exists a
  $\Gbf$-diagram $[\Zcal']$ containing an admissible fan
  $([X']\ot{}[Z']\to{}[U'])$, corresponding to the original admissible
  fan through the combinatorial isomorphism, such that, with the
  notations $\Xcal=\lc X\rc$ and $\Xcal'=\lc X'\rc$, the diagram
  $[\Zcal']$ satisfies
  \begin{enumerate}\def\theenumi{\roman{enumi}}
  \item 
    $\aikd([\Xcal'\rel U'],[\Xcal\rel U])\leq\epsilon$
  \item
    $\aikd(\Xcal',\Xcal)\leq \epsilon$
  \item
    $[Z'\rel X']\leq \epsilon$
  \end{enumerate}
\end{theorem}

The evaluation of entropy of an individual space in a tropical diagram
is a 1-Lipschitz linear functional, while the operation of
conditioning is also a Lipschitz map, see
\cite{Matveev-Conditioning-2019}. Thus in the settings of
Theorem~\ref{p:contraction} the following inequalities hold: for any
two spaces $[X_1]$ and $[X_2]$ in $[\Xcal]$ and corresponding spaces
$[X'_{1}]$ and $[X'_{2}]$ in $[\Xcal']$ the following inequalities
hold:
\begin{align*}
  &\Big|[X_1\rel U]-[X'_1\rel U']\Big|
   \leq 
   \epsilon\\
  &\Big|[X_1:X_2\rel U]-[X'_1:X'_2\rel U']\Big|
   \leq 
   2\epsilon\\
  &\Big|[X_1:U]-[X'_1:U']\Big|
   \leq 
   2\epsilon
   \tageq{entropy-preservation-epsilon}\\
  &\Big|[X_1:U\rel X_2]-[X'_1:U'\rel X'_2]\Big|
   \leq 3\epsilon\\
  &\Big|[U']-[X:U]\Big|\leq\epsilon
\end{align*}

The following much simpler theorem from~\cite{Matveev-Arrow-2019} is
the reverse of Theorem~\ref{p:contraction}.

\begin{theorem}{p:lifting}\emph{(Expansion)}
  Given a reduced admissible sub-fan $([X]\ot{}[Z']\to{}[U'])$ in a tropical
  $\Gbf$-diagram $\Zcal'$  and a non-negative
  number $\lambda\geq0$, there is another $\Gbf$-diagram $[\Zcal]$
  obtained from $\Zcal$ by arrow expansion such that $[Z\rel X]=\lambda$.
\end{theorem}

\subsection{Examples}
To illustrate the discussion above, we consider two examples of
admissible sub-fans and arrow contraction and expansion.

\subsubsection{}
As a first example, suppose we are given a tropical two-fan
$[\Zcal]=([X]\ot{}[Z]\to{}[U])$ as in Figure
\ref{fig:lowering-in-L2}. We may ask the following question:
\begin{quote}\em 
  {Can the mutual information between
  $[X]$ and $[U]$ be captured by a tropical space $[V]$? More
  precisely, is there a diamond extension
  \[
  \begin{cd}[row sep=-2mm,column sep=small]
    \mbox{}
    \&{}
    [Z]
    \arrow{dl}
    \arrow{dr}
    \&
    \mbox{}
    \\{}
    [X]
    \arrow{dr}
    \&
    \mbox{}
    \&{}
    [U]
    \arrow{dl}
    \\
    \mbox{}
    \&{}
    [V]
    \&
    \mbox{}
  \end{cd}
  \]
  such that $[V]=[X:U]$ (or equivalently $[X:U\rel V]=0)$?}
\end{quote}
The answer is, in general, no. However, by contracting and collapsing
the arrow $[Z]\to{}[X]$ we can still obtain a reduction
$([X]\to{}[V])$, where $[V]$ has the required ``size'', i.e.~its
entropy equals the mutual information between $[X]$ and $[U]$. If we
want to, we can still keep the spaces $[Z]$ and $[U]$ in the diagram
after contraction/collapse. Note, however, that in general there will
be no reduction $[U]\to{}[V]$ commuting with the other reductions.
\begin{figure}[h]
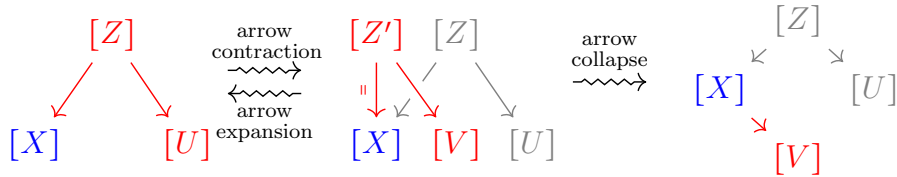

  \begin{tabular}{ccccc}
    \begin{cd}[column sep=0mm]
      \mbox{}
      \&{}
      \textcolor{red}{[Z]}
      \arrow[color=red]{dl}[above]{}%\textcolor{black}{\lambda}}
      \arrow[color=red]{dr}
      \&{}
      \\{}
      \textcolor{blue}{[X]}
      \&{}
      \&{}
      \textcolor{red}{[U]}
    \end{cd}
    &
    \hspace{-2em}
    \begin{cd}[row sep=0mm]
    \mbox{}
    \arrow[rightsquigarrow]{r}{%
      \parbox{14mm}{\tiny\center arrow\\[-0.6ex] contraction\\[0.5ex]}}
    \&
    \mbox{}
    \\
    \mbox{}
    \&
    \mbox{}
    \arrow[rightsquigarrow]{l}{%
      \parbox{14mm}{\tiny\center arrow\\[-0.6ex] expansion\\[0.5ex]}}
    \end{cd}
    \hspace{-1em}
    &
    \begin{cd}[column sep=0mm]
      \textcolor{red}{[Z']}
      \arrow[color=red]{d}[left]{\rotatebox{90}{$=$}}
      \&{}
      \textcolor{gray}{[Z]}
      \arrow[color=gray]{dl}
      \arrow[color=gray]{dr}
      \&{}
      \\{}
      \textcolor{blue}{[X]}
      \&{}
      \textcolor{red}{[V]}
      \arrow[from=ul,color=red,crossing over]{}
      \&{}
      \textcolor{gray}{[U]}
    \end{cd}
    &
    \hspace{-2em}
    \begin{cd}
    \mbox{}
    \arrow[rightsquigarrow]{r}{%
      \parbox{14mm}{\tiny\center arrow\\[-0.6ex] collapse\\[0.5ex]}}
    \&
    \mbox{}
    \end{cd}
    \hspace{-1em}
    &
    \begin{cd}[column sep=0mm,row sep=tiny]
      \mbox{}
      \&{}
      \textcolor{gray}{[Z]}
      \arrow[color=gray]{dl}
      \arrow[color=gray]{dr}
      \&{}
      \\{}
      \textcolor{blue}{[X]}
      \arrow[color=red]{dr}
      \&{}
      \&{}
      \textcolor{gray}{[U]}
      \\{}
      \&
      \textcolor{red}{[V]}
    \end{cd}    
  \end{tabular}
\caption{Contraction/expansion and arrow collapse in a two-fan}
\label{fig:lowering-in-L2}
\end{figure}

\subsubsection{}
As a second example, consider the tropical $\Lambdabf_{3}$-diagram
\[
  [\Zcal]=\<[X_{1}],[X_{2}],[U]\>
\]
shown in Figure \ref{fig:lowering-in-L3}.
Such examples can be particularly useful when the space $[U]$ is
chosen to satisfy additional properties. For instance, it could be
chosen such that the diagrams $[X_1]$ and $[X_2]$ are independent
conditioned on $[U]$. We will discuss such extensions elsewhere. 
The fan $([X]\ot{}[Z]\to{}[U])$ is admissible and the
ideal $\lc X\rc$ is the fan  $([X_{1}]\ot{}[X]\to{}[X_{2}])$.

If we contract $[Z]\to{}[U]$, we obtain a diagram with a new space
$[V]$, that has the same properties relative to $\lc X\rc=\<[X_1],
[X_2]\>$. The arrow expansion can be seen by reading the
picture backwards.

\begin{figure}
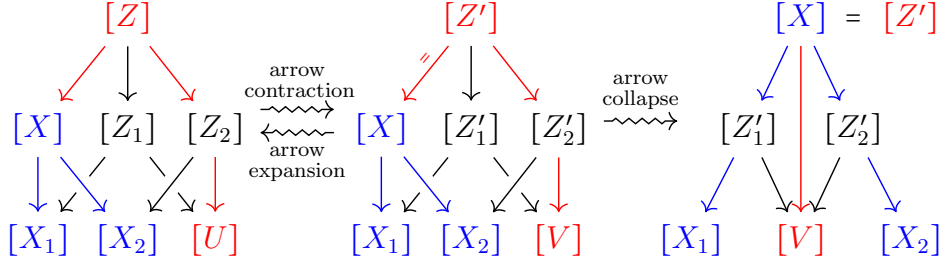

  \begin{tabular}{ccccc}
    \begin{cd}[column sep=0mm,row sep=normal]
      \mbox{}
      \&{}
      \textcolor{red}{[Z]}
      \arrow[color=red]{dl}
      \arrow{d}
      \arrow[color=red]{dr}
      \&{}
      \\{}
      \textcolor{blue}{[X]}
      \arrow[color=blue]{d}
      \&{}
      [Z_{1}]
      \arrow{dr}
      \arrow{dl}
      \&{}
      [Z_{2}]
      \arrow[color=red]{d}
      \\{}
      \textcolor{blue}{[X_{1}]}
      \&{}
      \textcolor{blue}{[X_{2}]}
      \arrow[from=ul, color=blue,crossing over]
      \arrow[from=ur, crossing over]
      \&{}
      \textcolor{red}{[U]}
      %\arrow[from=uul, crossing over, color=red]
    \end{cd}      
    &
    \hspace{-2em}
    \begin{cd}[row sep=0mm]
    \mbox{}
    \arrow[rightsquigarrow]{r}{%
            \parbox{14mm}{\tiny\center arrow\\[-0.6ex] contraction\\[0.5ex]}}
    \&
    \mbox{}
    \\
    \mbox{}
    \&
    \mbox{}
    \arrow[rightsquigarrow]{l}{%
      \parbox{14mm}{\tiny\center arrow\\[-0.6ex] expansion\\[0.5ex]}}
    \end{cd}
    \hspace{-2em}
    &
    \begin{cd}[column sep=0mm]
      \mbox{}
      \&{}
      \textcolor{red}{[Z']}
      \arrow[color=red]{dl}[above]{\rotatebox{50}{$=$}}
      \arrow{d}
      \arrow[color=red]{dr}
      \&{}
      \\{}
      \textcolor{blue}{[X]}
      \arrow[color=blue]{d}
      \&{}
      [Z'_{1}]
      \arrow{dr}
      \arrow{dl}
      \&{}
      [Z'_{2}]
      \arrow[color=red]{d}
      \\{}
      \textcolor{blue}{[X_{1}]}
      \&{}
      \textcolor{blue}{[X_{2}]}
      \arrow[from=ul, color=blue,crossing over]
      \arrow[from=ur, crossing over]
      \&{}
      \textcolor{red}{[V]}
      %\arrow[from=uul, crossing over, color=red]
    \end{cd}      
    &
    \hspace{-2em}
    \begin{cd}
    \mbox{}
    \arrow[rightsquigarrow]{r}{%
            \parbox{14mm}{\tiny\center arrow\\[-0.6ex] collapse\\[0.5ex]}}
    \&
    \mbox{}
    \end{cd}
    \hspace{-1em}
    &
    \hspace{-10mm}
    \begin{cd}[column sep=-4mm]
      \mbox{}
      \&{}
      \&{}
      \textcolor{blue}{[X]}
      \arrow[color=blue]{dl}
      \arrow[color=red]{dd}
      \arrow[color=blue]{dr}
      \&
      =
      \&{}
      \textcolor{red}{[Z']}
      \\{}
      \&{}
      [Z'_{1}]
      \arrow[color=blue]{dl}
      \arrow{dr}
      \&{}
      \&{}
      [Z'_{2}]
      \arrow[color=blue]{dr}
      \arrow{dl}
      \&{}
      \\{}
      \textcolor{blue}{[X_{1}]}
      \&{}
      \&{}
      \textcolor{red}{[V]}
      \&{}
      \&{}
      \textcolor{blue}{[X_{2}]}
    \end{cd}         
  \end{tabular}
  \caption{Arrow contraction and expansion in a $\Lambdabf_{3}$-diagram}
  \label{fig:lowering-in-L3}
\end{figure}

\section{Entropic Cone}
In this section we define the submodular, entropic and Abelian cones
associated to an indexing category $\Gbf$. 

\subsection{Vector-spaces $\Rbb^{\Gbf}$ and $\Rbb\otimes\Gbf$}
Given an indexing category $\Gbf$ we consider two linear
spaces associated to it. 
Recall that the vector space $\Rbb^{\Gbf}$
is the space of functions from $\Gbf$ to $\Rbb$. 
The second space, dual to the first one, is $\Rbb\otimes\Gbf$ -- the
vector-space of formal finite linear combinations of objects in $\Gbf$
with real coefficients. These two vector spaces are in natural duality
defined by
\[
  \text{For }
  f\in\Rbb^{\Gbf}
  \text{ and }
  \sum_{i\in\Gbf}\lambda_{i} \otimes i\in\Rbb\otimes\Gbf,
  \quad
  \<f,\sum_{i\in\Gbf}\lambda_{i}\otimes i\>
  :=
  \sum_{i\in\Gbf}\lambda_{i}f(i)
\]
The collection of vectors 
$\set{1 \otimes i}_{i \in \Gbf}=: 
\set{[i]}_{i \in \Gbf}$ 
forms a
basis of the space $\Rbb \otimes \Gbf$, and we denote the dual basis
in $\Rbb^{\Gbf}$ by $\set{f_i}_{i \in \Gbf}$. 
We also consider the following special vectors in $\Rbb\otimes\Gbf$:
\begin{itemize}
\item 
  The basis vectors $[i]:= 1\otimes i$.
\item 
  $[i|j]:=[\hat\imath]-[j]$, where $\hat\imath$ is
    the top object in a minimal fan $i\ot\hat\imath\to j$ in $\Gbf$.
\item 
  $[i:j]:=[i]+[j]-[\hat\imath]$, where $\hat\imath$ is top object in a
  minimal fan $i\ot \hat\imath\to j$ in $\Gbf$.
\item 
  $[i:j|k]:=[\hat\imath]+[\hat\jmath]-[k]-[l]$ where objects
  $\hat\imath$, $\hat\jmath$ and $[l]$ are included in the following
  minimal diagram in $\Gbf$
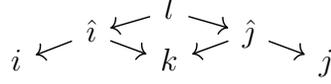
  \[
  \begin{cd}[row sep=-3mm,column sep=small]
    \mbox{} \& 
    \mbox{} \&  
    l
    \arrow{dl}
    \arrow{dr}
    \\
    \mbox{} \&
    \hat\imath
    \arrow{dl}
    \arrow{dr}
    \&
    \mbox{}
    \&
    \mbox{}
    \hat\jmath
    \arrow{dl}
    \arrow{dr}
    \\
    i
    \&\&
    k
    \&\&
    j
  \end{cd}
  \]
\end{itemize}
This notations are consistent with the notations for entropy and
mutual information, which we introduced in Section~\ref{s:entropy}, in
the sense that for a $\Gbf$-diagram $\Xcal=\set{X_{i};\chi_{ij}}$
holds
\begin{align*}
  \<\ent_{*}\Xcal\,,\, [i]\>&=[X_{i}]\\
  \<\ent_{*}\Xcal\,,\, [i|j]\>&=[X_{i}|X_{j}]\\
  \<\ent_{*}\Xcal\,,\, [i:j]\>&=[X_{i}:X_{j}]\\
  \<\ent_{*}\Xcal\,,\, [i:j|k]\>&=[X_{i}:X_{j}|X_{k}]\\
\end{align*}

\subsection{Submodular, entropic and Abelian cones}
Let $\Gbf=\set{i;\gamma_{ij}}$ be an indexing category. 
We define three closed, convex cones in $\Rbb^{\Gbf}$: the submodular cone
$\smc(\Gbf)$, the entropic cone $\ec(\Gbf)$ and the Abelian cone
$\abc(\Gbf)$. 

\subsubsection{The submodular cone}

The submodular cone $\smc(\Gbf)\subset\Rbb^{\Gbf}$ consists of nonnegative,
non-decreasing, submodular functions on the set of objects in the category
(points in the poset or vertices in the DAG) $\Gbf$. In essence, these
are the functions on $\Gbf$ that satisfy Shannon-like inequalities.
More formally it is defined as follows.

The properties nonnegativity, monotonicity and submodularity are
defined through linear inequalities. Every linear inequality for $f
\in \Rbb^\Gbf$ can be written in the form $\langle f, v \rangle \geq 0$ for some $v \in \mathbb{R} \otimes \Gbf$. A function $f \in
\Rbb^\Gbf$ is called

\begin{itemize}
\item 
  \emph{positive}, if $\< f,\, [i] \> \geq 0$ for every object
  $i \in \Gbf$
\item 
  \emph{monotone}, if $\< f,\, [i|j] \> \geq 0$, for every
  $i,j\in\Gbf$. 
\item 
  \emph{submodular}, if
  $\< f,\, [i:j] \> \geq 0$ and $\< f\,,\, [i:j|k] \> \geq 0$ for every
  $i,j,k\in\Gbf$ 
\end{itemize}
The submodular cone is dual to the cone spanned by Shannon-like
inequalities
\[
  \smc(\Gbf)
  :=
  \set{f\in\Rbb^{\Gbf}
    \st 
    \<f,v\>\geq0 \text{ for all }
    v\in\SH
  }
\]
where $\SH:=\set{[i],\,[i|j],\,[i:j],\,[i:j|k]\st i,j,k\in\Gbf}$.

\subsubsection{The entropic cone}
The entropic cone consists of functions on $\Gbf$ that are realizable
as entropies of tropical $\Gbf$-diagrams of probability spaces,
i.e. it is the image under the entropy map $\ent_{*}$ of the tropical
cone of minimal diagrams indexed by $\Gbf$:
\[
  \ec(\Gbf)
  :=
  \ent_{*}(\prob[\Gbf]_{\msf})
\]
In view of the tropical AEP Theorem \ref{p:aep-tropical} one can
equivalently define 
\[
  \ec(\Gbf)
  :=
  \operatorname{Closure}\big(\ent_{*}(\prob[\Gbf]_{\msf, \hsf})\big)
\]
where by $\prob[\Gbf]_{\msf, \hsf}$ we mean the space of minimal,
homogeneous, tropical $\Gbf$-diagrams.  As we explained in Section
\ref{se:full}, when $\Gbf = \Lambdabf_n$, diagrams correspond to
$n$-tuples of random variables. In this case, the entropic cone is
equal to the closure of the set of entropically representable vectors,
i.e. vectors whose coordinates are entropies of the $n$ random
variables and their joints, see \cite{Yeung-Information-2008}.

\subsubsection{The Abelian cone}
The Abelian cone consists of entropy vectors of Abe\-lian tropical diagrams
\[
  \abc(\Gbf)
  :=
  \ent_{*}(\prob[\Gbf]_{\mathsf{Ab}, \msf})
\]
The following two inclusions follow from the definitions and the fact that
entropy satisfies Shannon inequalities.
\[\tageq{inclusions}
\smc(\Gbf)\supset\ec(\Gbf)\supset\abc(\Gbf)
\]

\subsubsection{The cases of $\Gbf=\Lambdabf_{1}$, $\Lambdabf_{2}$, and
  $\Lambdabf_{3}$}
In this cases all three cones coincide. Essentially it means that any
tuple of numbers, that satisfy Shannon inequalities can be realized as
entropies of Abelian diagrams, see~\cite{Zhang-Characterization-1998}.

\subsection{The case $\Gbf=\mathbf{\Lambda}_{4}$}
The Zhang-Yeung non-Shannon information inequality
(\cite{Zhang-Characterization-1998}) shows that the submodular cone
$\smc(\Lambdabf_4)$ is strictly larger than the entropic cone
$\ec(\Lambdabf_4)$. It is also known that $\ec(\Lambdabf_4)$ is
strictly larger than $\abc(\Lambdabf_4)$, see for
example~\cite{Matus-Infinitely-2007}.  Hence, both inclusions
in~(\ref{eq:inclusions}) are proper.

The cone $\smc(\Lambdabf_{4})$ is polyhedral by definition, and it is
known that the cone $\abc(\Lambdabf_{4})$ is polyhedral as well, see,
for example,~\cite{Dougherty-Non-Shannon-2011}. In contrast, the
entropic cone $\ec(\Lambdabf_4)$ is not polyhedral, as has been shown
by Mat\'u\v s in \cite{Matus-Infinitely-2007}.

There are many upper and lower bounds for $\ec(\Lambdabf_{4})$. The
upper bounds are in the form of linear inequalities, some of them
organized in infinite families. A large list can be found
in~\cite{Dougherty-Non-Shannon-2011}. Lower bounds are in the form of
points in the complement
$\ec(\Lambdabf_{4})\setminus\abc(\Lambdabf_{4})$.

Note that there is an action of symmetric group $S_{4}$ on $\Lambdabf_{4}$,
$\prob[\Lambdabf_{4}]$, $\smc(\Lambdabf_{4})$, $\ec(\Lambdabf_{4})$
and $\abc(\Lambdabf_{4})$. 

We will adopt Mat\'u\v s' notations, where an integer (in small bold
face) represents the set of its decimal digits (eg
$\2\4\oto\set{2,4}\in\Lambdabf_{4}$).

\subsubsection{Ingleton inequalities and the Abelian cone $\abc(\Lambdabf_{4})$}
In addition to the Shannon inequalities, Abelian diagrams also satisfy
six Ingleton inequalities, corresponding to the Ingleton vector
\[
  \ing(\1\2;\3\4)
  :=
  -[\1:\2]+[\1:\2|\3]+[\1:\2|\4]+[\3:\4]\in\Rbb\otimes\Lambdabf_{4}
\]
and five other vectors obtained by permuting the coordinates.

The cone $\abc(\Lambdabf_{4})$ is a polyhedral cone dual to the cone
spanned by $\SH$ and six Ingleton vectors. Its structure is
well-known: it coincides with the cone called $\Hbf^\square$ in
\cite{Matus-Conditional-1995}. It has 35 extremal rays, grouped into ten
$S_{4}$-orbits.

\subsubsection{The submodular cone $\smc(\Lambdabf_4)$}
We will represent vectors in $\Rbb^{\Lambdabf_{4}}$ by writing their
coordinates in the following order
\[
\coords([\1],[\2],[\3],[\4];
[\1\2],[\1\3],[\1\4],[\2\3],[\2\4],[\3\4];
[\1\2\3],[\1\2\4],[\1\3\4],[\2\3\4];
[\1\2\3\4])
\]

The cone $\smc(\Lambdabf_4)$ has 41 extremal rays, grouped into eleven
$S_4$-orbits: the 35 rays that are extremal for $\abc(\Lambdabf_{4})$
and six special rays in the $S_{4}$-orbit of a ray generated by the
vector
\[
   \spc(\1\2;\3\4):=\coords(2,2,2,2;\;3,3,3,3,3,4;\;4,4,4,4;\;4)
\]

Note that $\<\spc(\1\2;\3\4)\,,\,\ing(\1\2;\3\4)\>=-1$. It is known that
$\spc(\1\2;\3\4)$ and the other special vectors are not in
$\ec(\Lambdabf_{4})$; they are neither representable as
entropy vectors of some diagram of probability spaces nor can they be
approximated by representable vectors.

\subsubsection{The non-Ingleton cone}
The closure of the complement
\[
  \smc(\Lambdabf_{4})\setminus\abc(\Lambdabf_{4})
\]
is the union of six cones with disjoint interiors, permuted by the
action of $S_{4}$. The stabilizer $D_{2}$ of this action is the
dihedral subgroup of $S_{4}$ preserving the partition
$\1\2\3\4=\1\2\cup\3\4$. It has order four and is isomorphic to
$\Zbb_{2}\times\Zbb_{2}$.

Consider one of these cones, containing $\spc(\1\2;\3\4)$ and denote
it by $\ning$. We will call it the non-Ingleton cone.  The cone
$\ning$ has a 14-dimensional simplex as a base. The vertices $a_1,
\dots, a_{15}$ and the dual faces $\alpha_1, \dots, \alpha_n$ of the
simplex are listed in Table~\ref{tbl:non-ingleton}.

\begin{table}
\[
\begin{array}{l|l|l|c}
  \text{Vertex}
  &
  \text{Dual face}
  &
  \text{Representative}
  &
  \text{$D_{2}$-orbit}
  \\\hline
  a_{1}=\coords(1,0,0,0;\;1,1,1,0,0,0;\;1,1,1,0;\;1)
  &
  \alpha_{1}=[\1|\2\3\4]
    &
  \begin{array}{l}
    l_{2}\cdot\big(
    \Zbb_{2};
    \set{0},\\\quad\;\;
    \Zbb_{2},
    \Zbb_{2},
    \Zbb_{2}\big)
  \end{array}
  &
  a_1,a_2
  \\
  a_{3}=\coords(0,0,1,0;\;0,1,0,1,0,1;\;0,1,1,1;\;1)
  &
  \alpha_{3}=[\3|\1\2\4]
    &
  \begin{array}{l}
    l_{2}\cdot\big(
    \Zbb_{2};
    \Zbb_{2},\\\quad\;\;
    \set{0},
    \Zbb_{2},
    \Zbb_{2}\big)
  \end{array}
  &
  a_3,a_4
  \\
  a_{5}=\coords(1,1,0,0;\;1,1,1,1,1,0;\;1,1,1,1;\;1)
  &
  \alpha_{5}=[\1:\3|\2]
  &
  \begin{array}{l}
    l_{2}\cdot\big(
    \Zbb_{2};
    \set{0},\\\quad\;\;
    \set{0},
    \Zbb_{2},
    \Zbb_{2}\big)
  \end{array}
  &
  \begin{array}{l}
    a_{5},a_{6},\\a_{7},a_{8}
  \end{array}
  \\
  a_{9}=\coords(1,1,1,0;\;1,1,1,1,1,1;\;1,1,1,1;\;1)
  &
  \alpha_{9}=[\1:\2|\4]
  &
  \begin{array}{l}
    l_{2}\cdot\big(
    \Zbb_{2};
    \set{0},\\\quad\;\;
    \set{0},
    \set{0},
    \Zbb_{2}\big)
  \end{array}
  &
  a_{9},a_{10}
  \\
  a_{11}=\coords(1,1,1,1;\;1,1,1,1,1,1;\;1,1,1,1;\;1)
  &
  \alpha_{11}=[\3:\4]
  &
  \begin{array}{l}
    l_{2}\cdot\big(
    \Zbb_{2};
    \set{0},\\\quad\;\;
    \set{0},
    \set{0},
    \set{0}\big)
  \end{array}
  &
  a_{11}
  \\
  a_{12}=\coords(1,0,1,1;\;1,2,2,1,1,2;\;2,2,2,2;\;2)
  &
  \alpha_{12}=[\3:\4|\1]
  &
  \begin{array}{l}
    l_{2}\cdot\big(
    (\Zbb_{2})^{2};\;
   \<\chi_{1}\>,\\\quad\;\;
   \<\chi_{1},\chi_{2}\>,
   \<\chi_{2}\>,\\\quad\;\;
   \<\chi_{1}+\chi_{2}\>\big)
  \end{array}
  &
  a_{12},a_{13} 
  \\
  a_{14}=\coords(1,1,1,1;\;2,2,2,2,2,2;\;3,3,3,3;\;3)
  &
  \alpha_{14}=[\1:\2|\3\4]
  &
  \begin{array}{l}
    l_{3}\cdot\big(
    (\Zbb_{3})^{3};
    \<\chi_{1},\chi_{2}\>,\\\quad\;\;
    \<\chi_{2},\chi_{3}\>,
    \<\chi_{3},\chi_{1}\>,\\\quad\;\;
    \<\chi_{1}+\chi_{2},\chi_{2}+\chi_{3}\>
    \big)
  \end{array}
  &
  a_{14}
  \\
  a_{15}=\coords(2,2,2,2;\;3,3,3,3,3,4;\;4,4,4,4;\;4)
  &
  \!\!\!\begin{array}{l}
  \alpha_{15}=\\\,\,\,\,\,-\ing(\1\2;\3\4)
  \end{array}
  &
  \text{Not representable}
  &
  a_{15}
  \\\hline
\end{array}
\]
\caption{The vertices and faces of the base simplex of non-Ingleton
  cone.
  \emph{The dihedral group $D_{2}$ acts on the simplex by transposing
  $\1$ and $\2$ and, independently, $\3$ and $\4$, so we list
  only one representative in each orbit. To shorten notations we
  set $l_{2}=(\ln2)^{-1}$ and $l_{3}:=(\ln3)^{-1}$. By
  $(\Zbb_{n})^{k}$ we mean the direct product of $k$ copies of the
  cyclic group of order $n$ and $\chi_{1},\ldots,\chi_{k}$ stand for
  the standard generators in $(\Zbb_{n})^{k}$.}}
\label{tbl:non-ingleton}
\end{table}

The covectors $\alpha_1, \dots, \alpha_{15}$ give convex coordinates
in the simplex.

\subsubsection{The cone $\ec(\Lambdabf_{4})$}
The cone $\ec(\Lambdabf_{4})$ is squeezed between $\abc$ and $\smc$
and the whole picture is $S_{4}$-symmetric. Thus the ``unknown'' part
of the $\ec(\Lambdabf_{4})$ is the intersection
$\ec':=\ec(\Lambdabf_{4})\cap\ning$. It contains the rays spanned by
  vectors $a_{1},\ldots,a_{14}$ and therefore the whole face
  $\set{\alpha_{15}=0}$.  The remaining part of the boundary
  $\partial_{+}\ec'$ is what we are after.  From convexity of $\ec'$
  it follows that this part of the boundary is the graph of a certain
  function defined on the cone spanned by $a_{1},\ldots,a_{14}$
\[
  \partial_{+}\ec'=\set{\alpha_{15}=\Phi(\alpha_{1},\ldots,\alpha_{14})}
\]
where $\Phi$ is defined by
\[
  \Phi(x_{1},\ldots,x_{14})
  :=
  \sup\set{\alpha_{15}(\xbf)
           \st
           \big(\alpha_{1}(\xbf),\ldots,\alpha_{14}(\xbf)\big)
           =(x_{1},\ldots,x_{14}),\;
           \xbf\in\ec(\Lambdabf_{4})}
\]
Obviously, the function $\Phi$ is 1-homogeneous.

\begin{theorem}{p:dim-reduction}
  The function $\Phi$ does not depend on the first four arguments.
\end{theorem}
\begin{proof}
For convenience, for a tropical $\Lambdabf_{4}$-diagram we write
\[
  A_{i}[\Xcal]:=\<\ent_{*}[\Xcal],\alpha_{i}\>
\]
e.g.  $A_{1}[\Xcal]=[X_{\1}|X_{\2\3\4}]$,
$A_{5}[\Xcal]=[X_{1}:X_{3}|X_{2}]$, etc.  Note that all $A_{i}$'s are
Lipschitz-continuous with respect to the input diagram with Lipschitz
constant at most $14$.
  
In terms of functionals $A_{i}$ the definition of the
function $\Phi$ can be rewritten as
\[
  \Phi(x_{1},\dots,x_{14})
  :=
  \sup\set{A_{15}[\Xcal] \st 
    A_{i}[\Xcal]=x_{i}\text{ for }
    1\leq i\leq 14;\;[\Xcal]\in\prob[\Lambdabf_{4}]_{\msf}}
\]

Consider a minimal tropical
$\Lambdabf_{4}$-diagram $[\Xcal]=\<[X_{\1}],[X_{\2}],[X_{\3}],[X_{\4}]\>$.
It contains an admissible sub-fan
$([X_{\2\3\4}]\ot{[X_{\1\2\3\4}]}\to{}[X_{\1}])$. 

Applying Theorem~\ref{p:contraction} to $[\Xcal]$ and parameter
$\epsilon>0$ we obtain another diagram $[\Xcal']$ such that
\begin{align*}
  &A_{1}[\Xcal']\leq\epsilon\\
  &\big|
    A_{i}[\Xcal']-A_{i}[\Xcal]
  \big|
  \leq 14\epsilon
  &&\text{for $i=2,\dots,15$}
\end{align*}

Repeatedly applying Theorem~\ref{p:contraction} to the resulting
diagram after circular permutation of terminal spaces we obtain a
$\Lambdabf_{4}$-diagram
\[
  [\Xcal'']=\<[X_{\1}''],[X''_{\2}],[X''_{\3}],[X''_{\4}]\>
\]
such that
\begin{align*}
  &A_{i}[\Xcal'']\leq(3\cdot14+1)\epsilon
  &&\text{for $i=1,2,3,4$}\\
  &\big|
    A_{i}[\Xcal'']-A_{i}[\Xcal]
  \big|
  \leq (4\cdot14)\epsilon
  &&\text{for $i=5,\dots,15$}
\end{align*}

Therefore, for any tuple $(x_{1},\ldots,x_{15})$ of
non-negative numbers there exists a tuple $(x_{1}'', \ldots, x_{15}'')$ such that
\begin{align*}
  &\Phi(x_{1},\ldots,x_{14})\leq\Phi(x''_{1},\ldots,x''_{14})+56\epsilon\\
  &x''_{i}\leq 43\epsilon 
  \text{ for $i=1,2,3,4$}\\
  &|x''_{i}-x_{i}|\leq 56\epsilon
  \text{ for $i=5,\ldots,14$}
\end{align*}

Since the function $\Phi$ is convex and therefore continuous, we can pass
to the limit with $\epsilon\to0$, obtaining the following result.
For any tuple $(x_{1},\ldots,x_{15})$ of non-negative numbers holds
\[
  \Phi(x_{1},\ldots,x_{14})\leq\Phi(0,0,0,0,x_{5},\ldots,x_{14})
\]

On the other hand, given a diagram $[\Ycal]$ with $A_{i}[\Ycal]=0$,
$i=1,2,3,4$, and a tuple of non-negative numbers
$(x_{1},x_{2},x_{3},x_{4})$, we can expand the arrows in the four
admissible fans, that we described above, to lengths
$(x_{1},x_{2},x_{3},x_{4})$. The resulting diagram $[\Ycal'']$
satisfies
\[
  A_{i}[\Ycal'']=
  \begin{cases}
    x_{i}&i=1,2,3,4\\
    A_{i}[\Ycal] & i=5,\ldots,15
  \end{cases}
\]
This implies 
\[
\Phi(x_{1},\ldots,x_{14})\geq\Phi(0,0,0,0,x_{5},\ldots,x_{14})
\]
for any non-negative $(x_{1},\ldots,x_{14})$.
\end{proof}

\bibliographystyle{alpha}       % APS-like style for physics
\bibliography{ReferencesKyb}
\end{document}